\documentclass[11pt]{article}
\usepackage[T1]{fontenc}
\usepackage[a4paper,margin=1in]{geometry}
\usepackage{amsmath,amssymb,amsthm,mathtools}
\usepackage{enumitem}
\usepackage{hyperref}
\usepackage{microtype}

\hypersetup{
  colorlinks=true,
  linkcolor=blue,
  citecolor=blue,
  urlcolor=blue
}

\newtheorem{theorem}{Theorem}[section]
\newtheorem{proposition}[theorem]{Proposition}
\newtheorem{definition}[theorem]{Definition}
\newtheorem{remark}[theorem]{Remark}
\newtheorem{example}[theorem]{Example}

\newtheorem{lemma}[theorem]{Lemma}

\title{From Bopp Shifts to Toroidal Shadows: K-Theoretic Gap Labels in Noncommutative Quantum Mechanics}
\author{S. Hasibul Hassan Chowdhury\\
  Department of Mathematics and Physical Sciences,\\
  BRAC University,\\
  Dhaka, Bangladesh\\
  \texttt{shhchowdhury@bracu.ac.bd}}
\date{}

\begin{document}

\maketitle

\begin{abstract}
We study Bopp shifts in two-dimensional noncommutative quantum mechanics (NCQM) through a functorial lens.  A nondegenerate NCQM sector with central character \((\hbar,\vartheta,B_{\mathrm{in}})\) determines a self-adjoint infinitesimal representation of the NCQM Lie algebra \(\mathfrak g_{\mathrm{NC}}\).  A Darboux normalization of its represented phase-space operators produces a self-adjoint infinitesimal representation of the Weyl--Heisenberg Lie algebra \(\mathfrak g_{\mathrm{WH}}\) with central parameter \(\hbar\), and hence defines a Bopp-shift functor that collapses the sector data
\[
(\hbar,\vartheta,B_{\mathrm{in}})\longmapsto \hbar .
\]
In particular, a generic NCQM sector is not equivalent, as a \(\mathfrak g_{\mathrm{NC}}\)-sector, to the ordinary QM sector viewed inside \(\mathfrak g_{\mathrm{NC}}\) with central character \((\hbar,0,0)\), even though their Bopp-shift images have the same Weyl--Heisenberg parameter.  To measure what this collapse forgets, we construct a toroidal shadow functor assigning to each periodic datum \(L=(a_x,a_y)\) and each NCQM sector with central character \(\chi=(\hbar,\vartheta,B_{\mathrm{in}})\) a phase-space noncommutative four-torus \(A^4_{\Theta_{\chi,L}}\).  Its \(K_0\)-trace pairing yields sector-sensitive gap labels whose top-degree coefficient is
\[
\frac{\vartheta B_{\mathrm{in}}-\hbar}{(2\pi)^2\hbar}.
\]
This coefficient is independent of the spatial cell area \(a_xa_y\) and equals the Pfaffian \(\operatorname{Pf}(\Theta_{\chi,L})\), the area-independent top-degree generator of the \(K_0\)-trace range.  Since strong Morita equivalence preserves this trace range up to positive scaling, non-proportionality of the trace ranges obstructs Morita equivalence of the shadows; this separates equal-\(\hbar\) sectors that \(\mathcal B_T\) identifies.  In the arithmetic subfamily where the normalized deformation coefficients \(\vartheta/A\) and \(B_{\mathrm{in}}A/\hbar\) are algebraic, the trace-range scale is forced to be trivial and \(|\operatorname{Pf}|\) becomes a computable separation criterion: \(|\vartheta B_{\mathrm{in}}-\hbar|\neq|\vartheta' B'_{\mathrm{in}}-\hbar|\) already implies inequivalence.  This is the sense in which the toroidal shadow functor is a sector-sensitive refinement of the Bopp-shift functor.
\end{abstract}

\section{Introduction}

Two-dimensional noncommutative quantum mechanics (NCQM) is often formulated
through operators $\hat X, \hat Y, \hat\Pi_x, \hat\Pi_y$ satisfying
\[
[\hat X,\hat Y]=i\vartheta \hat{I},\qquad [\hat\Pi_x,\hat\Pi_y]=i\hbar B_{\mathrm{in}}\hat{I},\qquad
[\hat X,\hat\Pi_x]=[\hat Y,\hat\Pi_y]=i\hbar \hat{I},
\]
with all other commutators vanishing. These operators are densely defined on the
Hilbert space $\mathcal{H}=L^2(\mathbb{R}^2)$; $\hat{I}$ denotes the identity operator
on $\mathcal{H}$. (Here and throughout, operators carrying a hat are represented on
$\mathcal{H}$; abstract Lie algebra elements have no hat.)

A common technique is to form new linear combinations of these operators, known
as generalized Bopp shifts, so that the resulting quadruple satisfies the
canonical commutation relations. Such transformations rewrite noncommutative
Hamiltonians in terms of auxiliary canonical operators and are therefore useful
for computations. They have also sometimes been interpreted as suggesting an
equivalence between NCQM and ordinary quantum mechanics (QM)
\cite{DjemaiSmail2004, DragovichRakic2004, Gamboa2001}.

The present paper examines what this passage really does. A Bopp-shifted
quadruple satisfies the defining relations of the Weyl--Heisenberg Lie algebra
$\mathfrak{g}_{\mathrm{WH}}$. The original operators, however, do not represent
$\mathfrak{g}_{\mathrm{WH}}$: as recalled below, they represent a larger
seven-dimensional step-two nilpotent Lie algebra $\mathfrak{g}_{\mathrm{NC}}$, the
Lie algebra of a triple central extension of the abelian group of translations in
$\mathbb{R}^4$ \cite{ChowdhuryAli2013, ChowdhuryAli2014, Chowdhury2017}. Its
regular self-adjoint infinitesimal irreducible representations are labelled by the central character
$\chi=(\hbar,\vartheta,B_{\mathrm{in}})$. Under the usual regularity convention,
these infinitesimal representations integrate to unitary irreducible
representations of the corresponding simply connected Lie group.

Thus a Bopp shift is not a unitary equivalence between two representations of
the same Lie algebra; it is an invertible linear transformation of represented
operators that maps a representation of $\mathfrak{g}_{\mathrm{NC}}$ to a
representation of a different Lie algebra, $\mathfrak{g}_{\mathrm{WH}}$. This
observation provides the starting point for the present paper. We interpret the
Bopp shift as a \emph{covariant functor} between regular infinitesimal
representation categories
$\operatorname{Rep}_{\mathrm{reg}}(\mathfrak{g}_{\mathrm{NC}})$ and
$\operatorname{Rep}_{\mathrm{reg}}(\mathfrak{g}_{\mathrm{WH}})$. Here
$\operatorname{Rep}_{\mathrm{reg}}(\mathfrak g)$ denotes the category of regular
infinitesimal representations realized by essentially self-adjoint represented
operators on a common invariant dense domain. Equivalently, after passing to the
self-adjoint closures of these represented operators and multiplying by $i$, one
obtains the usual skew-adjoint infinitesimal representation that integrates to a
unitary representation of the corresponding connected simply connected nilpotent
Lie group. Irreducibility is understood in this regular sense: the integrated
unitary representation is irreducible, equivalently the von Neumann algebra
generated by the one-parameter unitary groups obtained from the self-adjoint
closures of the represented operators acts irreducibly.

More precisely, after choosing Darboux normalizations, one obtains a Bopp-shift
functor
\[
\mathcal{B}_T:\mathcal{C}^{\circ}_{\mathfrak{g}_{\mathrm{NC}}}
\longrightarrow
\operatorname{Rep}_{\mathrm{reg}}(\mathfrak{g}_{\mathrm{WH}}),
\]
defined on the full subcategory of $\mathfrak{g}_{\mathrm{NC}}$-representations
whose irreducible constituents have nondegenerate central characters. Here
``nondegenerate'' means that the corresponding commutator form has full rank,
equivalently
\[
\hbar\neq0,\qquad \hbar-B_{\mathrm{in}}\vartheta\neq0.
\]
On objects, the functor sends the sector label
\[
(\hbar,\vartheta,B_{\mathrm{in}})
\quad\longmapsto\quad
\hbar.
\]
Thus the Bopp-shift passage retains the Planck parameter but forgets the two
additional NCQM central-character components. The ordinary QM sector, viewed inside $\mathfrak{g}_{\mathrm{NC}}$, is the nondegenerate sector with central character $(\hbar,0,0)$, while a generic NCQM sector, in the
terminology of this paper, means a nondegenerate sector
$(\hbar,\vartheta,B_{\mathrm{in}})$ with $\vartheta\neq0$ and
$B_{\mathrm{in}}\neq0$. The inequivalence between these sectors is an immediate
byproduct of central-character invariance; the main point here is to describe
the corresponding loss of sector data and to construct a refinement that
remembers part of it.

The Bopp-shift functor is useful for calculations because the target
Weyl--Heisenberg representation is, by the Stone--von Neumann theorem, unitarily
equivalent to the ordinary Schr\"odinger representation with parameter $\hbar$.
This is precisely why Bopp shifts can create the impression of equivalence:
one is comparing the auxiliary Weyl--Heisenberg representations, not the original
$\mathfrak{g}_{\mathrm{NC}}$-representations. For a thorough discussion of
Weyl--Heisenberg representation theory and the Stone--von Neumann theorem, see
Folland \cite{Folland1989}.

To retain part of the data discarded by $\mathcal{B}_T$, we construct a
\emph{toroidal shadow functor} $\mathcal{T}_L^\nabla$ associated with a spatial
periodicity datum $L=(a_x,a_y)$. For a nondegenerate sector $\chi$, we restrict
the Weyl system generated by $\hat X,\hat Y,\hat\Pi_x,\hat\Pi_y$ to a rank-four
phase-space lattice determined by $L$. After an intrinsic normalization, the
resulting unitaries generate a smooth noncommutative $4$-torus
$A^4_{\Theta_{\chi,L}}$. The dimensionless deformation matrix $\Theta_{\chi,L}$
depends on $\vartheta$ and $B_{\mathrm{in}}$ in a way that survives the
Bopp-shift collapse.

The target of $\mathcal{T}_L^\nabla$ is a Morita category of smooth
noncommutative $4$-tori arising from NCQM sectors: its objects are the tori
$A^4_{\Theta_{\chi,L}}$ and its morphisms are Morita equivalence bimodules. The
superscript $\nabla$ records the geometric use of these tori as receptacles for
finitely generated projective modules and Grassmann connections, which enter the
gap-label and Chern-character discussion below. This setting is directly motivated by condensed-matter
physics, where noncommutative tori appear naturally in the theory of the quantum
Hall effect and topological insulators \cite{Bellissard1994, ProdanSchulzBaldes2016}.
For any projection $p$ (for instance, a Fermi projection of a gapped
Hamiltonian), the trace pairing on $K_0(A^4_{\Theta_{\chi,L}})$ yields a
\emph{gap-label invariant} $\mathfrak{G}_{\chi,L}(p)$. The invariant takes the
form
\[
\mathfrak{G}_{\chi,L}(p) = N - \frac{\vartheta}{2\pi a_x a_y}C_{12}
- \frac{1}{2\pi}C_{13} - \frac{1}{2\pi}C_{24}
- \frac{B_{\mathrm{in}}a_x a_y}{2\pi\hbar}C_{34}
+ \frac{\vartheta B_{\mathrm{in}} - \hbar}{(2\pi)^2\hbar}Q,
\]
where $N, C_{ij}, Q$ are components of the $K_0$-class. Notably, the coefficient
of the top-degree component $Q$ is independent of the spatial cell area
$a_xa_y$ and is controlled by the intrinsic nondegeneracy combination
$\hbar-B_{\mathrm{in}}\vartheta$. Lower-degree coefficients, such as those of
$C_{12}$ and $C_{34}$, do depend on the cell area, reflecting how the periodic
lattice samples spatial noncommutativity and the magnetic or momentum
commutator. Explicit constructions of projective modules over noncommutative
tori, and their relevance to NCQM-generated noncommutative $4$-tori, are
provided by Rieffel's general theory \cite{Rieffel1988} and by the NCQM
construction in \cite{Chowdhury2017}.

The component $Q$ is the degree-four component of the $K_0$-class, the algebraic
analogue of the second Chern number in four-dimensional quantum Hall systems
\cite{ZhangHu2001}. Because the four toroidal directions are compactified
position- and momentum-like NCQM directions rather than an ordinary Brillouin
torus, we treat $Q$ here strictly as a $K$-theoretic gap-label component; a
current-response or pumping interpretation is not claimed
(Section~\ref{sec:examples-interpretation}).

These constructions yield the paper's main results, which make precise the sense
in which the toroidal shadow remembers what the Bopp shift forgets. The trace
pairing endows each shadow with a trace range
$\Gamma_{\chi,L}=\tau_*\bigl(K_0(A^4_{\Theta_{\chi,L}})\bigr)\subset\mathbb R$,
which strong Morita equivalence preserves only up to a positive scalar. Hence
non-proportionality of two trace ranges obstructs Morita equivalence of the
corresponding shadows, and already separates equal-$\hbar$ sectors that
$\mathcal B_T$ identifies (Theorem~\ref{thm:trace-range}). This obstruction
becomes fully explicit in the arithmetic subfamily where the dimensionless
coefficients $\vartheta/a_xa_y$ and $B_{\mathrm{in}}a_xa_y/\hbar$ are algebraic:
there the transcendence of $1/2\pi$ forces the Morita scaling factor to be
trivial (Lemma~\ref{lem:lambda-one}), and the area-independent Pfaffian
$\operatorname{Pf}(\Theta_{\chi,L})=(\vartheta B_{\mathrm{in}}-\hbar)/((2\pi)^2\hbar)$
becomes a \emph{computable} separation criterion: the shadows are not Morita
equivalent whenever
\[
|\vartheta B_{\mathrm{in}}-\hbar|\neq|\vartheta' B'_{\mathrm{in}}-\hbar|
\]
(Theorem~\ref{thm:pfaffian-separation}), with the exact proportionality condition
given in Proposition~\ref{prop:exact-morita}. Thus the single intrinsic
combination $\hbar-B_{\mathrm{in}}\vartheta$ that governs nondegeneracy and fixes
the degree-four gap-label coefficient also controls the Morita obstruction --- a
computable invariant, detectable in the $K$-theory of the toroidal shadow, that
survives the Bopp-shift collapse.

The paper is organised as follows. Section~\ref{sec:sectors-bopp} fixes the
representation-theoretic setting and distinguishes Bopp shifts from unitary
equivalence. Section~\ref{sec:bopp-functor} records the Bopp-shift collapse.
Section~\ref{sec:toroidal-functor} constructs the toroidal shadow, and
Section~\ref{sec:k0-q} derives the gap-label formula and proves the separation
results (Theorems~\ref{thm:trace-range} and~\ref{thm:pfaffian-separation}). Section~\ref{sec:examples-interpretation}
gives examples and discusses the degree-four component and its physical
interpretation.

\section{NCQM sectors, central characters, and Bopp shifts}\label{sec:sectors-bopp}

This section fixes the representation-theoretic setting: the Lie algebra
$\mathfrak{g}_{\mathrm{NC}}$, its central characters, the nondegeneracy
condition, and the distinction between linear transformations of represented
operators (Bopp shifts, Darboux normalizations) and unitary equivalence. The
central character $(\hbar,\vartheta,B_{\mathrm{in}})$ is a unitary invariant of
a sector by Schur's lemma; it is the datum the Bopp-shift functor will forget.

Throughout, the representations of $\mathfrak{g}_{\mathrm{NC}}$ are the regular
infinitesimal representations described in the Introduction: the noncentral
generators act by essentially self-adjoint operators on a common invariant dense
domain in the Hilbert space $\mathcal{H}=L^2(\mathbb{R}^2)$, and the associated
skew-adjoint infinitesimal representation integrates to a unitary representation
of the connected simply connected nilpotent Lie group $G_{\mathrm{NC}}$,
irreducibility being understood in that integrated sense. The identity operator
on $\mathcal{H}$ is denoted $\hat{I}$; abstract Lie algebra elements carry no
hat, while their represented self-adjoint counterparts do.

Let $\mathfrak{g}_{\mathrm{NC}}$ be the seven-dimensional step-two nilpotent
Lie algebra generated by four noncentral elements
\[
X,\;Y,\;\Pi_x,\;\Pi_y
\]
and three central elements
\[
Z_\hbar,\;Z_\vartheta,\;Z_B.
\]
The nonzero brackets in $\mathfrak{g}_{\mathrm{NC}}$ are
\[
[X,\Pi_x]=Z_\hbar,\qquad [Y,\Pi_y]=Z_\hbar,
\]
\[
[X,Y]=Z_\vartheta,\qquad [\Pi_x,\Pi_y]=Z_B.
\]

If $\pi$ is an irreducible regular self-adjoint infinitesimal representation of $\mathfrak{g}_{\mathrm{NC}}$, then the center acts by scalars:
\begin{equation}\label{eq:center-action}
\pi(Z_\hbar)=i\hbar \hat{I},\qquad
\pi(Z_\vartheta)=i\vartheta \hat{I},\qquad
\pi(Z_B)=i\hbar B_{\mathrm{in}}\hat{I}.
\end{equation}
The triple
\[
\chi=(\hbar,\vartheta,B_{\mathrm{in}})
\]
is the central character of the sector. The represented noncentral
operators (denoted $\hat X,\hat Y,\hat\Pi_x,\hat\Pi_y$) satisfy
\[
[\hat X,\hat Y]=i\vartheta \hat{I},\qquad
[\hat\Pi_x,\hat\Pi_y]=i\hbar B_{\mathrm{in}}\hat{I},
\]
\[
[\hat X,\hat\Pi_x]=[\hat Y,\hat\Pi_y]=i\hbar \hat{I},
\]
with the remaining mixed commutators equal to zero.

Let
\[
V=\operatorname{span}\{\hat X,\hat Y,\hat\Pi_x,\hat\Pi_y\}.
\]
The central character defines the skew-symmetric form
\[
\Omega_\chi(U,W)=\chi([U,W])
\]
on $V$. In the ordered basis
\[
(\hat X,\hat Y,\hat\Pi_x,\hat\Pi_y),
\]
its matrix is
\begin{equation}\label{eq:sigma-chi}
\Sigma_\chi=
\begin{bmatrix}
0 & \vartheta & \hbar & 0\\
-\vartheta & 0 & 0 & \hbar\\
-\hbar & 0 & 0 & \hbar B_{\mathrm{in}}\\
0 & -\hbar & -\hbar B_{\mathrm{in}} & 0
\end{bmatrix}.
\end{equation}
A direct calculation gives
\[
\det \Sigma_\chi
=
\hbar^2(\hbar-B_{\mathrm{in}}\vartheta)^2.
\]

\begin{definition}[Nondegenerate central character]\label{def:nondegenerate}
A central character $\chi = (\hbar,\vartheta,B_{\mathrm{in}})$ is called
\emph{nondegenerate} if it satisfies
\[
\hbar \neq 0 \qquad\text{and}\qquad \hbar - B_{\mathrm{in}}\vartheta \neq 0.
\]
A $\mathfrak{g}_{\mathrm{NC}}$-sector with a nondegenerate central character is called a
\emph{nondegenerate sector}.
\end{definition}

Ordinary two-dimensional quantum mechanics (QM) appears inside the regular
self-adjoint representation theory of $\mathfrak{g}_{\mathrm{NC}}$ as the Weyl--Heisenberg sector
with central character $(\hbar,0,0)$, where the additional central directions $Z_\vartheta$ and
$Z_B$ act trivially. This sector is nondegenerate because
$\hbar\neq0$ and $\hbar - 0\cdot0 = \hbar \neq 0$.

A \emph{nondegenerate NCQM sector} is any $\mathfrak{g}_{\mathrm{NC}}$-sector of the form
$(\hbar,\vartheta,B_{\mathrm{in}})$ whose central character is nondegenerate
(Definition~\ref{def:nondegenerate}). In this paper, by a \emph{generic NCQM
sector} we mean a nondegenerate NCQM sector with
$\vartheta\neq0$ and $B_{\mathrm{in}}\neq0$. Since the scalar action of the
center is a unitary invariant of the sector, such a generic sector is
\emph{not} unitarily equivalent, as a $\mathfrak{g}_{\mathrm{NC}}$-representation,
to the ordinary QM sector $(\hbar,0,0)$. The issue addressed below is different: Bopp shifts and Darboux
normalizations can put the noncentral commutator matrix into canonical form,
and this can create the appearance of equivalence unless the full central
character is kept in view.

\subsection{Generalized Bopp shifts, Darboux normalization, and unitary equivalence}

We now separate two notions that the rest of the paper keeps carefully apart:
unitary changes of representation and linear transformations of represented
operators.

\begin{definition}[Unitary change of representation]
A unitary change of representation is implemented by a unitary operator $U$.
For represented operators one obtains
\[
\hat A\longmapsto U\hat A U^{-1}.
\]
In particular, for a represented quadruple
\[
\hat Z=(\hat X,\hat Y,\hat\Pi_x,\hat\Pi_y)^T,
\]
a unitary change gives
\[
\hat Z\longmapsto U\hat Z U^{-1}
=(U\hat X U^{-1},U\hat Y U^{-1},U\hat\Pi_x U^{-1},U\hat\Pi_y U^{-1})^T.
\]
Such a transformation preserves all commutators and the scalar action of the
center.
\end{definition}

\begin{definition}[Linear transformation of represented operators]
Let
\[
\hat Z=(\hat X,\hat Y,\hat\Pi_x,\hat\Pi_y)^T
\]
be a quadruple of represented operators on a common invariant dense domain. A
linear transformation of represented operators is a new quadruple
\[
\hat Z' = M\hat Z,\qquad M\in GL(4,\mathbb R).
\]
Equivalently, $\hat Z'_j = \sum_k M_{jk} \hat Z_k$. Such a transformation is a unitary
change of representation only if there exists a unitary operator $U$ such that
$\hat Z'_j = U \hat Z_j U^{-1}$ for all $j$; this is not automatic.
\end{definition}

A generalized Bopp shift is an invertible linear transformation of this kind,
usually used to express noncommuting phase-space operators in terms of
auxiliary canonical operators, or conversely. Concrete NCQM realizations,
including the two-parameter families used in the group-theoretic formulation
of NCQM (see Example~\ref{ex:rs-family}), should be understood in this sense:
they provide linear realizations of the noncentral operators in a fixed
representation sector. A Bopp shift is therefore not a unitary equivalence in
$\operatorname{Rep}_{\mathrm{reg}}(\mathfrak{g}_{\mathrm{NC}})$ nor in
$\operatorname{Rep}_{\mathrm{reg}}(\mathfrak{g}_{\mathrm{WH}})$; its proper
categorical interpretation (as a functor) will be given in
Section~\ref{sec:bopp-functor}.

\begin{definition}[Canonical quadruple of represented operators]
A canonical quadruple with Planck parameter $\hbar$ is a quadruple of
represented self-adjoint operators
\[
\hat Z_{\mathrm{can}}=(\hat x,\hat y,\hat p_x,\hat p_y)^T
\]
on a common invariant dense domain satisfying
\[
[\hat x,\hat p_x]=[\hat y,\hat p_y]=i\hbar \hat{I},
\qquad
[\hat x,\hat y]=[\hat p_x,\hat p_y]=[\hat x,\hat p_y]=[\hat y,\hat p_x]=0.
\]
Its commutation matrix is
\begin{equation}\label{eq:sigma-can}
\Sigma_\hbar^{\mathrm{can}} =
\begin{bmatrix}
0&0&\hbar&0\\
0&0&0&\hbar\\
-\hbar&0&0&0\\
0&-\hbar&0&0
\end{bmatrix}.
\end{equation}
\end{definition}

\begin{definition}[Generalized Bopp shift and Darboux normalization]
Let
\[
\hat Z = (\hat X,\hat Y,\hat\Pi_x,\hat\Pi_y)^T
\]
be a quadruple intended to realize the NCQM commutation relations associated
with a central character $(\hbar,\vartheta,B_{\mathrm{in}})$. A generalized
Bopp shift is an invertible real linear transformation
\[
\hat Z = M \hat Z_{\mathrm{can}},\qquad M\in GL(4,\mathbb R),
\]
such that the components of $\hat Z$ satisfy the prescribed NCQM commutation
relations. Equivalently,
\[
M \Sigma_\hbar^{\mathrm{can}} M^T = \Sigma_\chi,
\]
where $\Sigma_\chi$ is the matrix in \eqref{eq:sigma-chi}.
The inverse transformation,
\[
\hat Z_{\mathrm{can}} = T \hat Z,\quad T = M^{-1},
\]
is called a Darboux normalization of the NCQM quadruple. The nondegeneracy
condition (Definition~\ref{def:nondegenerate}) is exactly the condition under
which such a Darboux normalization exists.
\end{definition}

Being linear transformations of the represented operators rather than unitary
conjugations, \(M\) and \(T=M^{-1}\) merely put the noncentral commutator matrix
into canonical Darboux form; neither is a unitary equivalence of
\(\mathfrak g_{\mathrm{NC}}\)-sectors.

\begin{example}[Coordinate Bopp shift]

Assume \(B_{\mathrm{in}}=0\). Let \[ \hat Z_{\mathrm{can}} = (\hat x,\hat y,\hat p_x,\hat p_y)^T \] be a canonical quadruple of represented self-adjoint operators satisfying \[ [\hat x,\hat p_x]=[\hat y,\hat p_y]=i\hbar \hat I, \] with all other commutators equal to zero. Define a new represented quadruple \[ \hat Z = (\hat X,\hat Y,\hat\Pi_x,\hat\Pi_y)^T \] by \[ \hat X=\hat x-\frac{\vartheta}{2\hbar}\hat p_y,\qquad \hat Y=\hat y+\frac{\vartheta}{2\hbar}\hat p_x, \] \[ \hat\Pi_x=\hat p_x,\qquad \hat\Pi_y=\hat p_y. \] A direct computation gives \[ [\hat X,\hat Y]=i\vartheta\hat I, \qquad [\hat X,\hat\Pi_x]=[\hat Y,\hat\Pi_y]=i\hbar\hat I, \qquad [\hat\Pi_x,\hat\Pi_y]=0, \] with the remaining mixed commutators equal to zero. Thus the quadruple \[ (\hat X,\hat Y,\hat\Pi_x,\hat\Pi_y) \] realizes the NCQM commutation relations with central character \[ (\hbar,\vartheta,0). \] This construction is an invertible linear transformation of represented operators, not a unitary conjugation of the canonical quadruple. Indeed, a unitary conjugation would preserve the commutator \[ [\hat x,\hat y]=0. \] However, the transformed represented operators satisfy \[ [\hat X,\hat Y]=i\vartheta\hat I, \] which is nonzero when \(\vartheta\neq0\).
\end{example}

\begin{example}[The $(r,s)$-family as a full phase-space realization]\label{ex:rs-family}

The two-parameter \((r,s)\)-family used in the group-theoretic formulation of NCQM gives a family of concrete represented quadruples for one fixed nondegenerate central character. Let \[ \chi=(\hbar,\vartheta,B_{\mathrm{in}}) \] be nondegenerate, so that \[ \hbar\neq0, \qquad \hbar-B_{\mathrm{in}}\vartheta\neq0. \] For each admissible pair \((r,s)\), the represented quadruple \[ \hat Z_{r,s} = (\hat X^{r,s},\hat Y^{r,s}, \hat\Pi_x^{r,s},\hat\Pi_y^{r,s})^T \] is obtained from an auxiliary canonical quadruple \[ \hat Z_{\mathrm{can}} = (\hat x,\hat y,\hat p_x,\hat p_y)^T \] by an invertible real linear transformation \[ \hat Z_{r,s}=M_{r,s}\hat Z_{\mathrm{can}}, \qquad M_{r,s}\in GL(4,\mathbb R). \] The matrix \(M_{r,s}\) is chosen so that \[ M_{r,s}\Sigma_{\hbar}^{\mathrm{can}}M_{r,s}^T = \Sigma_\chi. \] Equivalently, the represented operators \[ \hat X^{r,s},\quad \hat Y^{r,s},\quad \hat\Pi_x^{r,s},\quad \hat\Pi_y^{r,s} \] satisfy \[ [\hat X^{r,s},\hat Y^{r,s}] = i\vartheta\hat I, \qquad [\hat\Pi_x^{r,s},\hat\Pi_y^{r,s}] = i\hbar B_{\mathrm{in}}\hat I, \] \[ [\hat X^{r,s},\hat\Pi_x^{r,s}] = [\hat Y^{r,s},\hat\Pi_y^{r,s}] = i\hbar\hat I, \] with the remaining mixed commutators equal to zero. For fixed \(\chi\), different admissible choices \((r,s)\) and \((r',s')\) give different concrete realizations of the same \(\mathfrak g_{\mathrm{NC}}\)-sector; since the central character is unchanged, the corresponding regular irreducible representations are unitarily equivalent as representations of \(\mathfrak{g}_{\mathrm{NC}}\). The inverse matrix \[ M_{r,s}^{-1} \] has a different role. It sends the represented NCQM quadruple \[ \hat Z_{r,s} = (\hat X^{r,s},\hat Y^{r,s}, \hat\Pi_x^{r,s},\hat\Pi_y^{r,s})^T \] to an auxiliary canonical quadruple \[ \hat Z_{\mathrm{can}} = M_{r,s}^{-1}\hat Z_{r,s}, \] which satisfies the ordinary Weyl--Heisenberg commutation relations with central parameter \(\hbar\) and is therefore naturally regarded as a represented quadruple for the Weyl--Heisenberg Lie algebra \(\mathfrak g_{\mathrm{WH}}\), not as one carrying the full \(\mathfrak g_{\mathrm{NC}}\)-central character \[ (\hbar,\vartheta,B_{\mathrm{in}}). \] The explicit matrices \(M_{r,s}\) are given in the group-theoretic construction of the NCQM representations \cite{Chowdhury2017,ChowdhuryChowdhuryDuha2021}.
\end{example}

In particular, by the Stone--von Neumann theorem the Weyl--Heisenberg
representation obtained from a Darboux normalization is determined up to unitary
equivalence by $\hbar$ alone, so the data $\vartheta$ and $B_{\mathrm{in}}$ are
not retained after normalization. This loss is formalized in the next section.

\section{The Bopp-shift functor and object-level collapse}\label{sec:bopp-functor}

We record the Bopp-shift passage as a functor. The point is that the loss of
$\vartheta$ and $B_{\mathrm{in}}$ occurs already on objects: the functor
replaces $(\hbar,\vartheta,B_{\mathrm{in}})$ by the Weyl--Heisenberg parameter
$\hbar$.

Let $\mathcal{C}^{\circ}_{\mathfrak{g}_{\mathrm{NC}}}$ denote the full subcategory of
$\operatorname{Rep}_{\mathrm{reg}}(\mathfrak{g}_{\mathrm{NC}})$ whose objects are finite direct sums of
irreducible regular self-adjoint infinitesimal representations with nondegenerate central characters:
\[
\hbar\neq 0,
\qquad
\hbar-B_{\mathrm{in}}\vartheta\neq 0.
\]
Morphisms are bounded $\mathfrak{g}_{\mathrm{NC}}$-intertwiners, i.e. bounded operators that preserve the relevant common invariant domains and intertwine the represented infinitesimal operators there.

For each nondegenerate central character
$\chi=(\hbar,\vartheta,B_{\mathrm{in}})$, choose a Darboux normalization
$T_\chi\in GL(4,\mathbb R)$ such that
\begin{equation}\label{eq:darboux-normalization}
T_\chi\Sigma_\chi T_\chi^T=\Sigma_\hbar^{\mathrm{can}},
\end{equation}
where
\begin{equation}\label{eq:sigma-can-bopp}
\Sigma_\hbar^{\mathrm{can}}=
\begin{bmatrix}
0&0&\hbar&0\\
0&0&0&\hbar\\
-\hbar&0&0&0\\
0&-\hbar&0&0
\end{bmatrix}.
\end{equation}
For an irreducible representation $\pi_\chi$, define
\[
\hat Z_{\pi_\chi}=
(\hat X_{\pi_\chi},\hat Y_{\pi_\chi},\hat\Pi_{x,\pi_\chi},\hat\Pi_{y,\pi_\chi})^T
\]
and
\[
\hat Z'_{\pi_\chi}=T_\chi \hat Z_{\pi_\chi}.
\]
The transformed operators satisfy the ordinary Weyl--Heisenberg commutation
relations with Planck parameter $\hbar$. Thus, at the infinitesimal level, they define an object of $\operatorname{Rep}_{\mathrm{reg}}(\mathfrak{g}_{\mathrm{WH}})$; under the regularity assumptions, their exponentials integrate to a unitary representation of the simply connected Weyl--Heisenberg group.

On each irreducible summand the Darboux-normalized operators satisfy the
Weyl--Heisenberg relations with parameter $\hbar$, and a
$\mathfrak g_{\mathrm{NC}}$-intertwiner $A$ intertwines the transformed operators
as well (they are fixed linear combinations of the originals); on finite direct
sums one proceeds summandwise. Hence
\[
\mathcal{B}_T:\mathcal{C}^{\circ}_{\mathfrak{g}_{\mathrm{NC}}}\longrightarrow \operatorname{Rep}_{\mathrm{reg}}(\mathfrak{g}_{\mathrm{WH}}),
\qquad
(\hbar,\vartheta,B_{\mathrm{in}})\longmapsto \hbar,
\]
is a covariant functor. (It depends on the choice of Darboux family
$\{T_\chi\}$ only up to natural isomorphism: two choices differ by a
canonical-symplectic $R_\chi=S_\chi T_\chi^{-1}$, metaplectically implemented by
a unitary, giving a natural transformation $\mathcal B_T\Rightarrow\mathcal B_S$.)

The essential point is that the collapse happens \emph{on objects}: the label
$(\hbar,\vartheta,B_{\mathrm{in}})$ is sent to $\hbar$ alone. Thus for
$\vartheta\neq0$ or $B_{\mathrm{in}}\neq0$ the sectors
$\pi_{\hbar,\vartheta,B_{\mathrm{in}}}$ and $\pi_{\hbar,0,0}$ are non-isomorphic
in $\mathcal C^\circ_{\mathfrak g_{\mathrm{NC}}}$ (their central characters
differ, so
$\operatorname{Hom}(\pi_{\hbar,\vartheta,B_{\mathrm{in}}},\pi_{\hbar,0,0})=0$ by
Schur's lemma), yet their $\mathcal B_T$-images coincide up to unitary
equivalence. The functor is therefore neither an equivalence nor full; but the
loss of $\vartheta,B_{\mathrm{in}}$ is already complete at the object level,
independently of these facts. This collapse motivates the toroidal shadow
construction below.

\section{The toroidal shadow functor}\label{sec:toroidal-functor}\label{sec:toroidal-shadow}

This section constructs the positive replacement for the data lost under the
Bopp-shift functor.  For a fixed spatial periodicity datum $L=(a_x,a_y)$, a
nondegenerate NCQM sector determines dimensionless compact Weyl multipliers,
hence a noncommutative $4$-torus.  We then define the target category, the
toroidal shadow functor, and the precise sense in which this functor separates
some sectors identified by the Bopp shift.

Let $\chi=(\hbar,\vartheta,B_{\mathrm{in}})$ be a nondegenerate central character.
Suppose the system is placed in a background with spatial periods
\[
L=(a_x,a_y).
\]
Define dimensionless phase-space operators
\begin{equation}\label{eq:v-def}
\hat v_1=\frac{\hat X}{a_x},
\qquad
\hat v_2=\frac{\hat Y}{a_y},
\qquad
\hat v_3=\frac{a_x\hat\Pi_x}{\hbar},
\qquad
\hat v_4=\frac{a_y\hat\Pi_y}{\hbar}.
\end{equation}
Set
\[
U_j=e^{i\hat v_j},
\qquad
j=1,2,3,4.
\]
With the convention \(U_j=e^{i\hat v_j}\), the compactification is imposed only
at the Weyl-unitary level. More precisely, for each \(j\) we use the central
\(2\pi\)-shift relation
\[
\hat v_j\sim \hat v_j+2\pi n\hat I,\qquad n\in\mathbb{Z},
\]
these shifts leaving the Weyl unitary unchanged,
\[
e^{i(\hat v_j+2\pi n\hat I)}=e^{i\hat v_j}.
\]
Transporting this convention through \eqref{eq:v-def}, and writing the case
\(n=1\), gives
\[
\hat X\sim \hat X+2\pi a_x\hat I,\qquad
\hat Y\sim \hat Y+2\pi a_y\hat I,
\]
\[
\hat\Pi_x\sim \hat\Pi_x+\frac{2\pi\hbar}{a_x}\hat I,\qquad
\hat\Pi_y\sim \hat\Pi_y+\frac{2\pi\hbar}{a_y}\hat I.
\]
This notation does not identify unbounded self-adjoint operators and does not
define a quotient of the \(*\)-algebra generated by
\(\hat X,\hat Y,\hat\Pi_x,\hat\Pi_y,\hat I\). We also do not take \(\sim\) to be
the full equivalence relation given by equality of exponentials of arbitrary
self-adjoint operators; for general \(A,B\), the equality \(e^{iA}=e^{iB}\) need
not imply \(A-B\in 2\pi\mathbb{Z}\hat I\). Only the explicit central shifts above
are intended. Thus the unbounded operators serve only as infinitesimal
scaffolding for the bounded Weyl unitaries \(U_j=e^{i\hat v_j}\), and only these
unitaries descend to the toroidal shadow algebra \(A^4_{\Theta_{\chi,L}}\).
For example, the conjugate pair \(\hat X,\hat\Pi_x\) gives
\[
[\hat v_1,\hat v_3]=[\hat X,\hat\Pi_x]/\hbar=i,
\qquad
U_1U_3=e^{-i}U_3U_1,
\]
which yields the universal mixed entry
\((\Theta_{\chi,L})_{13}=-1/(2\pi)\) in \eqref{eq:theta-entries} below.

Since the commutators are central, the Baker--Campbell--Hausdorff formula gives
\begin{equation}\label{eq:torus-relations}
U_iU_j=e^{2\pi i(\Theta_{\chi,L})_{ij}}U_jU_i.
\end{equation}
With the ordering
\[
1=\hat X,
\qquad
2=\hat Y,
\qquad
3=\hat\Pi_x,
\qquad
4=\hat\Pi_y,
\]
the nonzero entries of the skew matrix $\Theta_{\chi,L}$ are
\begin{equation}\label{eq:theta-entries}
(\Theta_{\chi,L})_{12}=-\frac{\vartheta}{2\pi a_xa_y},
\qquad
(\Theta_{\chi,L})_{13}=-\frac{1}{2\pi},
\end{equation}
\begin{equation}\label{eq:theta-entries2}
(\Theta_{\chi,L})_{24}=-\frac{1}{2\pi},
\qquad
(\Theta_{\chi,L})_{34}=-\frac{B_{\mathrm{in}}a_xa_y}{2\pi\hbar},
\end{equation}
and
\[
(\Theta_{\chi,L})_{14}=(\Theta_{\chi,L})_{23}=0.
\]
Every entry of $\Theta_{\chi,L}$ is dimensionless. The spatial periods enter
only through the physical area $A=a_xa_y$. The mixed position--momentum entries
are universal because they arise from the canonical commutators after the
intrinsic normalization \eqref{eq:v-def}.

\begin{definition}[Toroidal shadow algebra]\label{def:toroidal-shadow-algebra}
The toroidal shadow algebra associated with $(\chi,L)$ is the smooth
noncommutative $4$-torus $A^4_{\Theta_{\chi,L}}$ generated by unitaries
$U_1,U_2,U_3,U_4$ satisfying
\[
U_iU_j=e^{2\pi i(\Theta_{\chi,L})_{ij}}U_jU_i.
\]
\end{definition}

This algebra is the phase-space toroidal shadow of the $\mathfrak{g}_{\mathrm{NC}}$-sector.
It is not an invariant of $\chi$ alone; it is an invariant of the pair
$(\chi,L)$.

\subsection{The toroidal shadow functor and its target}\label{sec:morita}

For a smooth noncommutative $4$-torus $A^4_\Theta$ we write
$\delta^{\Theta}_1,\ldots,\delta^{\Theta}_4$ for the canonical derivations,
$\delta^{\Theta}_j(U_k)=2\pi i\,\delta_{jk}U_k$. The natural target for the
shadow is a category of noncommutative tori in which isomorphism is Morita
equivalence, since Morita-equivalent tori have identified $K$-theory and
gap-label data.

\begin{definition}[Target category]\label{def:target}
Fix $L=(a_x,a_y)$. Let $\operatorname{Tor}_{\mathfrak g_{\mathrm{NC}},L}$ be the
category whose objects are the smooth tori $A^4_{\Theta_{\chi,L}}$ arising from
nondegenerate $\mathfrak g_{\mathrm{NC}}$-sectors, and whose morphisms
$A^4_\Theta\to A^4_{\Theta'}$ are (isomorphism classes of)
$(A^4_\Theta,A^4_{\Theta'})$-Morita equivalence bimodules, composed by balanced
tensor product. Two objects are isomorphic in
$\operatorname{Tor}_{\mathfrak g_{\mathrm{NC}},L}$ precisely when the tori are
Morita equivalent; we write $[\Theta_{\chi,L}]_{\mathrm{Mor}}$ for the
isomorphism class. Let
$\operatorname{Add}(\operatorname{Tor}_{\mathfrak g_{\mathrm{NC}},L})$ be its
finite additive completion.
\end{definition}

Explicit finitely generated projective modules and constant-curvature
connections over the four-tori arising from irreducible
$\mathfrak g_{\mathrm{NC}}$-sectors were constructed by Rieffel's method in
\cite{Chowdhury2017,Rieffel1988}; these supply the geometric receptacle for the
$K$-theoretic gap labels of Section~\ref{sec:k0-q} and, for a gapped
Hamiltonian, the Grassmann connection of a Fermi projection provides the
curvature data entering the Chern character.

\begin{definition}[Toroidal shadow functor]\label{def:toroidal-shadow-functor}
For fixed $L=(a_x,a_y)$, define
\[
\mathcal T_L^{\nabla}:
\mathcal{C}^{\circ}_{\mathfrak{g}_{\mathrm{NC}}}
\longrightarrow
\operatorname{Add}\left(\operatorname{Tor}_{\mathfrak{g}_{\mathrm{NC}},L}\right)
\]
on an irreducible object by $\pi_\chi\mapsto A^4_{\Theta_{\chi,L}}$, on finite
direct sums summandwise, and on morphisms by sending the scalar intertwiner
matrices of $\mathcal C^\circ_{\mathfrak g_{\mathrm{NC}}}$ (all Hom-spaces
between distinct sectors vanish by Schur, and endomorphisms of an isotypic
component are scalar matrices) to the same scalar matrices on the corresponding
free modules. This is a covariant $\mathbb C$-linear functor.
\end{definition}

The content of $\mathcal T_L^{\nabla}$ is entirely at the object level, and the
next remark makes precise the qualitative contrast with $\mathcal B_T$.

\begin{remark}[Collapse versus separation]\label{rem:collapse-separation}
The two functors fail to be equivalences in opposite directions. On objects,
$\mathcal B_T$ \emph{collapses}: it sends $(\hbar,\vartheta,B_{\mathrm{in}})$ to
$\hbar$, identifying all equal-$\hbar$ sectors. By contrast
$\mathcal T_L^{\nabla}$ \emph{separates}: it sends $\chi$ to
$A^4_{\Theta_{\chi,L}}$, whose deformation matrix retains $\vartheta$ and
$B_{\mathrm{in}}$. Consequently $\mathcal T_L^{\nabla}$ cannot factor through
$\mathcal B_T$. A dual symptom is that $\mathcal T_L^{\nabla}$ maps the
\emph{zero} Hom-space between two distinct sectors to a possibly \emph{nonzero}
Morita Hom-space between their shadows; the toroidal target thus carries
relational data absent from
$\operatorname{Rep}_{\mathrm{reg}}(\mathfrak g_{\mathrm{NC}})$. Whether that
Hom-space is nonzero for a given pair --- equivalently, whether the two shadows
are Morita equivalent --- is a concrete arithmetic question.
Theorem~\ref{thm:trace-range} gives the general trace-range obstruction to such
Morita equivalence, while Theorem~\ref{thm:pfaffian-separation} gives a
computable Pfaffian criterion in the algebraic-parameter subfamily.
\end{remark}

\begin{example}[Integral phase shifts]\label{ex:integral-shift}
If $N$ is an integral skew-symmetric $4\times4$ matrix then
$e^{2\pi i(\Theta+N)_{ij}}=e^{2\pi i\Theta_{ij}}$ for all $i,j$, so
$A^4_{\Theta+N}\cong A^4_\Theta$. Thus the toroidal phase is well defined modulo
integral skew-symmetric matrices, and Morita equivalence of higher-dimensional
tori is governed more generally by the arithmetic action of
\cite{RieffelSchwarz1999,Li2004}. In particular, distinct concrete matrices
$\Theta_{\chi,L}\neq\Theta_{\chi',L}$ need not give distinct Morita classes,
which is why the separation criterion below is stated at the level of a Morita
invariant.
\end{example}

\section{\texorpdfstring{$K_0$}{K0}-gap labels and the intrinsic \texorpdfstring{$Q$}{Q}-coefficient}\label{sec:k0-q}\label{sec:k0}

This section extracts the numerical invariant from the toroidal shadow.  The
trace pairing on $K_0(A^4_{\Theta_{\chi,L}})$ gives a dimensionless gap-label
formula.  We then isolate the coefficient of the top-degree component $Q$ and
show that it is independent of the spatial cell area $a_xa_y$, even though the
full toroidal shadow still depends on the periodicity datum $L$.

For a smooth noncommutative $4$-torus,
\[
K_0(A^4_\Theta)\cong \Lambda^{\mathrm{even}}\mathbb Z^4\cong \mathbb Z^8.
\]
We write a $K_0$-class as
\[
[p]=(N;C_{12},C_{13},C_{14},C_{23},C_{24},C_{34};Q),
\]
where $N$ is the degree-zero component, $C_{ij}$ are degree-two Chern
components, and $Q$ is the degree-four component.

For a projection $p\in M_n(A^4_{\Theta_{\chi,L}})$, define the toroidal gap
label
\begin{equation}\label{eq:gap-label-def}
\mathfrak G_{\chi,L}(p)=\tau_*([p]).
\end{equation}
The trace pairing has the form
\begin{equation}\label{eq:trace-pairing-general}
\tau_*([p])
=
N+\sum_{i<j}(\Theta_{\chi,L})_{ij}C_{ij}
+\operatorname{Pf}(\Theta_{\chi,L})Q.
\end{equation}
Since
\[
(\Theta_{\chi,L})_{14}=(\Theta_{\chi,L})_{23}=0,
\]
we have
\[
\operatorname{Pf}(\Theta_{\chi,L})
=
(\Theta_{\chi,L})_{12}(\Theta_{\chi,L})_{34}
-
(\Theta_{\chi,L})_{13}(\Theta_{\chi,L})_{24}.
\]
Substituting the entries \eqref{eq:theta-entries}--\eqref{eq:theta-entries2}
gives
\begin{equation}\label{eq:gap-label-formula}
\mathfrak G_{\chi,L}(p)
=
N
-
\frac{\vartheta}{2\pi a_xa_y}C_{12}
-
\frac{1}{2\pi}C_{13}
-
\frac{1}{2\pi}C_{24}
-
\frac{B_{\mathrm{in}}a_xa_y}{2\pi\hbar}C_{34}
+
\frac{\vartheta B_{\mathrm{in}}-\hbar}{(2\pi)^2\hbar}Q.
\end{equation}
This is the sector-sensitive gap-label formula associated with the toroidal
shadow.

If $p=P_F$ is the Fermi projection of a gapped Hamiltonian in a matrix algebra
over $A^4_{\Theta_{\chi,L}}$, then $\mathfrak G_{\chi,L}(P_F)$ has the
interpretation of a normalized integrated density of states below the
corresponding spectral gap.

\subsection{The intrinsic nature of the \texorpdfstring{$Q$}{Q}-coefficient}\label{sec:q}

The full toroidal shadow depends on the spatial periodicity datum $L$. This
dependence is expected: the construction compactifies the Weyl system using a
chosen spatial lattice. However, the coefficient of the degree-four component
$Q$ is independent of the spatial cell area.

\begin{proposition}[Intrinsic top-degree coefficient]\label{prop:q-coefficient}
Let $A=a_xa_y$. For the intrinsic normalization \eqref{eq:v-def},
\[
\operatorname{Pf}(\Theta_{\chi,L})
=
\frac{\vartheta B_{\mathrm{in}}-\hbar}{(2\pi)^2\hbar}.
\]
In particular, the coefficient of $Q$ in the gap-label formula is independent
of $A=a_xa_y$.
\end{proposition}

\begin{proof}
We have
\[
(\Theta_{\chi,L})_{12}
=
-\frac{\vartheta}{2\pi A},
\qquad
(\Theta_{\chi,L})_{34}
=
-\frac{B_{\mathrm{in}}A}{2\pi\hbar},
\]
and
\[
(\Theta_{\chi,L})_{13}
=
(\Theta_{\chi,L})_{24}
=
-\frac{1}{2\pi}.
\]
Therefore
\[
\operatorname{Pf}(\Theta_{\chi,L})
=
\left(-\frac{\vartheta}{2\pi A}\right)
\left(-\frac{B_{\mathrm{in}}A}{2\pi\hbar}\right)
-
\left(-\frac{1}{2\pi}\right)
\left(-\frac{1}{2\pi}\right).
\]
Thus
\[
\operatorname{Pf}(\Theta_{\chi,L})
=
\frac{\vartheta B_{\mathrm{in}}}{(2\pi)^2\hbar}
-
\frac{1}{(2\pi)^2}
=
\frac{\vartheta B_{\mathrm{in}}-\hbar}{(2\pi)^2\hbar}.
\]
\end{proof}

This result has a clear interpretation. The lower-degree $C_{12}$ and
$C_{34}$ terms depend on the spatial cell area:
\[
-\frac{\vartheta}{2\pi A}C_{12},
\qquad
-\frac{B_{\mathrm{in}}A}{2\pi\hbar}C_{34}.
\]
They measure how the spatial noncommutativity and magnetic/momentum commutator
are sampled by the chosen periodic cell. By contrast, the $Q$-coefficient is
independent of $A$. It is controlled by the intrinsic nondegeneracy combination
\[
\hbar-B_{\mathrm{in}}\vartheta.
\]
Hence the $Q$-term is not a lattice-cell-area effect. If a gapped Hamiltonian
realizes a Fermi projection with $Q(P_F)\neq 0$, then the degree-four
contribution to the gap label is governed by sector-level phase-space data
rather than by the size of the spatial periodic cell.

\subsection{A separation criterion from the trace range}\label{sec:separation-criterion}

We now make the qualitative refinement of
Remark~\ref{rem:collapse-separation} precise. For a smooth noncommutative
$4$-torus the range of the canonical trace on $K_0$,
\begin{equation}\label{eq:trace-range}
\tau_*\bigl(K_0(A^4_\Theta)\bigr)\subset\mathbb R,
\end{equation}
is the additive subgroup generated by the $2\pi$-normalized Pfaffians of the
even principal submatrices of $\Theta$: in dimension four, by $1$, the entries
$\Theta_{ij}$ $(i<j)$, and the $4\times4$ Pfaffian $\operatorname{Pf}(\Theta)$.
This is Elliott's trace-range computation, recalled in the normalization we use
by Prodan and Schulz-Baldes \cite[Sec.~5.7, Thm.~5.7.1 and Cor.~5.7.2]{ProdanSchulzBaldes2016}
(with the convention $U_iU_j=e^{2\pi i\Theta_{ij}}U_jU_i$, i.e.\ $B=2\pi\Theta$);
see also \cite{Elliott1984,RieffelSchwarz1999}. Strong Morita equivalence
preserves this range \emph{as an additive subgroup of $\mathbb R$, up to a single
positive scalar} (the dimension of the equivalence bimodule); it does not, in
general, preserve the individual named generators. The correct obstruction to
Morita equivalence is therefore non-proportionality of the two trace ranges.

Throughout this subsection fix $L=(a_x,a_y)$, write $A=a_xa_y$, and let
$\chi=(\hbar,\vartheta,B_{\mathrm{in}})$,
$\chi'=(\hbar,\vartheta',B'_{\mathrm{in}})$ be nondegenerate central characters
with the \emph{same} Planck parameter $\hbar$ (this is the family that
$\mathcal B_T$ collapses). By
\eqref{eq:theta-entries}--\eqref{eq:theta-entries2} and
Proposition~\ref{prop:q-coefficient} the trace range is, explicitly,
\begin{equation}\label{eq:trace-range-explicit}
\Gamma_{\chi,L}
:=
\tau_*\bigl(K_0(A^4_{\Theta_{\chi,L}})\bigr)
=
\Bigl\langle\,
1,\;
-\tfrac{\vartheta}{2\pi A},\;
-\tfrac1{2\pi},\;
-\tfrac{B_{\mathrm{in}}A}{2\pi\hbar},\;
\tfrac{\vartheta B_{\mathrm{in}}-\hbar}{(2\pi)^2\hbar}
\,\Bigr\rangle_{\mathbb Z},
\end{equation}
and similarly $\Gamma_{\chi',L}$ with $(\vartheta',B'_{\mathrm{in}})$.

The general obstruction, valid for arbitrary real NCQM parameters, is the
following.

\begin{theorem}[Trace-range separation]\label{thm:trace-range}
If there is no positive real $\lambda$ with
$\Gamma_{\chi',L}=\lambda\,\Gamma_{\chi,L}$, then $A^4_{\Theta_{\chi,L}}$ and
$A^4_{\Theta_{\chi',L}}$ are not strongly Morita equivalent, so
$[\Theta_{\chi,L}]_{\mathrm{Mor}}\neq[\Theta_{\chi',L}]_{\mathrm{Mor}}$ and hence
$\mathcal T_L^{\nabla}(\pi_\chi)\not\cong\mathcal T_L^{\nabla}(\pi_{\chi'})$,
whereas $\mathcal B_T(\pi_\chi)\cong\mathcal B_T(\pi_{\chi'})$. In particular such
sectors are separated by the toroidal shadow functor although identified by the
Bopp-shift functor.
\end{theorem}

\begin{proof}
Strong Morita equivalence induces a $K_0$-isomorphism carrying the trace pairing
to the trace pairing scaled by the positive dimension of the equivalence
bimodule; hence proportionality of the trace ranges is necessary for Morita
equivalence, and its failure is an obstruction. The Bopp-shift images are
irreducible Weyl--Heisenberg representations with the same central parameter
$\hbar$, hence unitarily equivalent by the Stone--von Neumann theorem.
\end{proof}

The Pfaffian is the area-independent, top-degree generator of $\Gamma_{\chi,L}$;
in the arithmetic subfamily where the dimensionless normalized parameters are
algebraic, its absolute value becomes a computable form of this obstruction. We
make this precise. Note that $\vartheta/A$ and $B_{\mathrm{in}}A/\hbar$ are (up to
the universal factor $2\pi$) the deformation-matrix entries $\Theta_{12}$ and
$\Theta_{34}$ themselves; requiring them algebraic is the standard setting of
algebraic deformation parameters in the noncommutative-torus literature.

\begin{lemma}[Rigidity of the scale in the algebraic subfamily]\label{lem:lambda-one}
Let $K\subset\mathbb R$ be a subfield over which $t:=\tfrac1{2\pi}$ is
transcendental, and suppose
\[
\frac{\vartheta}{A},\ \frac{B_{\mathrm{in}}A}{\hbar},\
\frac{\vartheta'}{A},\ \frac{B'_{\mathrm{in}}A}{\hbar}\ \in K.
\]
Then $\Gamma_{\chi,L}$ and $\Gamma_{\chi',L}$ are contained in the graded
$K$-vector space $V_K:=K\oplus K\,t\oplus K\,t^2$, and any $\lambda>0$ with
$\Gamma_{\chi',L}=\lambda\,\Gamma_{\chi,L}$ satisfies $\lambda=1$.
\end{lemma}

\begin{proof}
Since $(\vartheta B_{\mathrm{in}}-\hbar)/\hbar=(\vartheta/A)(B_{\mathrm{in}}A/\hbar)-1\in K$,
all five generators of $\Gamma_{\chi,L}$ in \eqref{eq:trace-range-explicit} lie in
$V_K$, with degree-$0$ part $\langle1\rangle$, degree-$1$ part spanned by
$\vartheta/A,\,1,\,B_{\mathrm{in}}A/\hbar$ (times $t$), and degree-$2$ part
$(\vartheta B_{\mathrm{in}}-\hbar)/\hbar$ (times $t^2$); the same holds for
$\Gamma_{\chi',L}$. By nondegeneracy the degree-$2$ part is nonzero and the
degree-$0$ part is $\langle1\rangle$, so the $K$-span of each trace range is all
of $V_K$. Hence $\lambda V_K=V_K$, i.e.\ multiplication by $\lambda$ is a
$K$-linear automorphism of $V_K$. Because $t$ is transcendental over $K$, the
powers $\{t^n:n\ge0\}$ are $K$-linearly independent. Writing $\lambda=a+bt+ct^2$
with $a,b,c\in K$ (as $\lambda=\lambda\cdot1\in V_K$), the requirement $\lambda
t^2=at^2+bt^3+ct^4\in V_K$ forces the coefficients of $t^3$ and $t^4$ to vanish,
i.e.\ $b=c=0$; thus $\lambda=a\in K$. A scalar $\lambda\in K$ preserves the
grading; in degree $0$ both ranges have graded part exactly
$\langle1\rangle_{\mathbb Z}=\mathbb Z$, so $\lambda\,\mathbb Z=\mathbb Z$, whence
$\lambda=1$.
\end{proof}

Under this hypothesis proportionality is equality, and equality may be tested one
graded degree at a time. The degree-$2$ (Pfaffian) part gives at once the
following computable criterion.

\begin{theorem}[Pfaffian separation criterion, algebraic subfamily]\label{thm:pfaffian-separation}
Suppose the normalized parameters
$\vartheta/A,\ B_{\mathrm{in}}A/\hbar,\ \vartheta'/A,\ B'_{\mathrm{in}}A/\hbar$
lie in a subfield $K\subset\mathbb R$ over which $t=1/(2\pi)$ is transcendental
(for instance $K=\mathbb Q^{\mathrm{alg}}\cap\mathbb R$, where $\mathbb Q^{\mathrm{alg}}$
is the algebraic closure of $\mathbb Q$ in $\mathbb C$, so that
$\mathbb Q^{\mathrm{alg}}\cap\mathbb R$ is the field of real algebraic numbers).
If
\begin{equation}\label{eq:pf-sep}
\bigl|\operatorname{Pf}(\Theta_{\chi,L})\bigr|
\neq
\bigl|\operatorname{Pf}(\Theta_{\chi',L})\bigr|,
\qquad\text{equivalently}\qquad
|\vartheta B_{\mathrm{in}}-\hbar|\neq|\vartheta' B'_{\mathrm{in}}-\hbar|,
\end{equation}
then $A^4_{\Theta_{\chi,L}}$ and $A^4_{\Theta_{\chi',L}}$ are not strongly Morita
equivalent, so
$[\Theta_{\chi,L}]_{\mathrm{Mor}}\neq[\Theta_{\chi',L}]_{\mathrm{Mor}}$ and hence
$\mathcal T_L^{\nabla}(\pi_\chi)\not\cong\mathcal T_L^{\nabla}(\pi_{\chi'})$,
whereas $\mathcal B_T(\pi_\chi)\cong\mathcal B_T(\pi_{\chi'})$. In particular the
two sectors are separated by the toroidal shadow functor although identified by
the Bopp-shift functor.
\end{theorem}

\begin{proof}
Suppose the shadows were strongly Morita equivalent. By
Lemma~\ref{lem:lambda-one} (whose hypothesis holds) the scale is $\lambda=1$, so
$\Gamma_{\chi',L}=\Gamma_{\chi,L}$ as subgroups of $V_K$, and in particular their
degree-$2$ parts coincide. The degree-$2$ part of $\Gamma_{\chi,L}$ is the cyclic
group $\langle(\vartheta B_{\mathrm{in}}-\hbar)/\hbar\rangle_{\mathbb Z}\,t^2$;
two cyclic subgroups $\langle x\rangle_{\mathbb Z}$, $\langle y\rangle_{\mathbb Z}$
of $K$ coincide iff $x=\pm y$. Hence $\vartheta B_{\mathrm{in}}-\hbar
=\pm(\vartheta' B'_{\mathrm{in}}-\hbar)$, i.e.
$|\vartheta B_{\mathrm{in}}-\hbar|=|\vartheta' B'_{\mathrm{in}}-\hbar|$,
contradicting \eqref{eq:pf-sep}. Since
$\operatorname{Pf}(\Theta_{\chi,L})=(\vartheta B_{\mathrm{in}}-\hbar)/((2\pi)^2\hbar)$
by Proposition~\ref{prop:q-coefficient}, the two forms of \eqref{eq:pf-sep} are
equivalent. Finally $\mathcal B_T(\pi_\chi)$ and $\mathcal B_T(\pi_{\chi'})$ are
irreducible Weyl--Heisenberg representations with the same central parameter
$\hbar$, hence unitarily equivalent by the Stone--von Neumann theorem.
\end{proof}

Thus, in the algebraic subfamily, the absolute value of the Pfaffian --- the
area-independent, top-degree generator of the trace range, and the coefficient of
the degree-four component $Q$ in the gap-label formula
\eqref{eq:gap-label-formula} --- is a computable Morita obstruction: whenever it
differs between two such equal-$\hbar$ sectors, the toroidal shadows are
inequivalent. The same combination $\hbar-\vartheta B_{\mathrm{in}}$ that fixes
the nondegeneracy of the sector and controls the degree-four gap label therefore
also separates Bopp-collapsed sectors. For arbitrary real parameters the general
obstruction of Theorem~\ref{thm:trace-range} applies instead.

For completeness we record, in the same algebraic subfamily, the exact condition
for the remaining case $|\operatorname{Pf}(\Theta_{\chi,L})|=|\operatorname{Pf}(\Theta_{\chi',L})|$,
where the degree-one lattice must also be compared.

\begin{proposition}[Exact proportionality criterion, algebraic subfamily]\label{prop:exact-morita}
Under the hypothesis of Theorem~\ref{thm:pfaffian-separation}, let
\[
D_1(\chi):=\Bigl\langle\,1,\ \tfrac{\vartheta}{A},\ \tfrac{B_{\mathrm{in}}A}{\hbar}\,\Bigr\rangle_{\mathbb Z}\subset K
\]
be the degree-one lattice of $\chi$. Then $\Gamma_{\chi,L}$ and $\Gamma_{\chi',L}$
are proportional (equivalently, equal) if and only if
\[
\vartheta B_{\mathrm{in}}-\hbar=\pm(\vartheta' B'_{\mathrm{in}}-\hbar)
\qquad\text{and}\qquad
D_1(\chi)=D_1(\chi').
\]
Consequently the shadows are \emph{not} Morita equivalent whenever either
condition fails.
\end{proposition}

\begin{proof}
By Lemma~\ref{lem:lambda-one} proportionality is equality with $\lambda=1$, and
by the $K$-independence of $1,t,t^2$ equality holds iff it holds in each graded
degree. Degree $0$ is automatic ($\mathbb Z=\mathbb Z$). Degree $2$ is the cyclic
condition $\vartheta B_{\mathrm{in}}-\hbar=\pm(\vartheta' B'_{\mathrm{in}}-\hbar)$
as in the proof of Theorem~\ref{thm:pfaffian-separation}. Degree $1$ is the
equality of the subgroups $\langle-\vartheta/A,-1,-B_{\mathrm{in}}A/\hbar\rangle_{\mathbb Z}
=D_1(\chi)$, since the sign of each generator is immaterial for the group they
generate. Combining the three degrees gives the stated criterion.
\end{proof}

\begin{remark}[Scope]\label{rem:criterion-scope}
The general separation statement is Theorem~\ref{thm:trace-range}, valid for
arbitrary real NCQM parameters: non-proportionality of the trace ranges obstructs
Morita equivalence. Theorem~\ref{thm:pfaffian-separation} and
Proposition~\ref{prop:exact-morita} refine this to a computable, sector-level
criterion --- a one-line check on $|\vartheta B_{\mathrm{in}}-\hbar|$ --- under
the arithmetic hypothesis that the dimensionless normalized coefficients
$\vartheta/A$, $B_{\mathrm{in}}A/\hbar$ lie in a field $K$ over which $1/2\pi$ is
transcendental. This hypothesis is natural from the arithmetic viewpoint on
noncommutative tori --- $\vartheta/A$ and $B_{\mathrm{in}}A/\hbar$ are, up to the
factor $2\pi$, the deformation coefficients entering $\Theta_{12},\Theta_{34}$, and
requiring them algebraic is the standard algebraic-deformation setting, with the
rational case $\vartheta/A, B_{\mathrm{in}}A/\hbar\in\mathbb Q$ as the elementary
special case --- but it is nonetheless an additional restriction on the real NCQM
parameter space. The unrestricted statement remains
Theorem~\ref{thm:trace-range}, which uses the full trace range and is valid for
arbitrary real parameters. The rigidity $\lambda=1$
(Lemma~\ref{lem:lambda-one}) then uses only nondegeneracy and the transcendence
of $1/2\pi$ over $K$, together with the fact that the unit, the universal mixed
entry $-1/2\pi$, and the nonzero Pfaffian populate three distinct powers of
$1/2\pi$. We do not address Morita equivalence \emph{across} different
$\hbar$, nor the finer $\mathrm{SO}(4,4\,|\,\mathbb Z)$ arithmetic needed to
classify the shadows completely \cite{RieffelSchwarz1999,Li2004}; these are not
required for the separation statement, whose purpose is to distinguish sectors
that $\mathcal B_T$ identifies.
\end{remark}

\section{Examples, model classes, and physical interpretation}\label{sec:examples-interpretation}\label{sec:examples}

We record examples interpreting the gap-label formula --- the magnetic Brillouin
subalgebra, a lower-degree Powers--Rieffel projection, and a finite-range model
class for the degree-four component --- and separate what is established from
what is conditional.

\subsection{Magnetic Brillouin subalgebra}\label{subsec:magnetic-brillouin}

In the magnetic sector $(\hbar,0,B_{\mathrm{in}})$, the spatial noncommutativity
parameter $\vartheta$ vanishes, but the momentum commutator remains nontrivial:
$[\hat\Pi_x,\hat\Pi_y]=i\hbar B_{\mathrm{in}}\hat{I}$. The magnetic translation operators are
defined as
\[
T_1 = \exp\!\left(-\frac{i}{\hbar}a_x\hat\Pi_x\right),\qquad
T_2 = \exp\!\left(-\frac{i}{\hbar}a_y\hat\Pi_y\right),
\]
where $a_x,a_y$ are the spatial periods. These unitaries satisfy the commutation
relation
\begin{equation}\label{eq:mag-translation-comm}
T_1 T_2 = e^{2\pi i\alpha_B}\,T_2 T_1,\qquad
\alpha_B = -\frac{B_{\mathrm{in}}a_x a_y}{2\pi\hbar}\mod 1.
\end{equation}
The $C^*$-algebra generated by $T_1,T_2$ is the \emph{magnetic noncommutative
Brillouin torus}; it is isomorphic to the noncommutative $2$-torus with
deformation parameter $\alpha_B$.

A standard Hamiltonian in this subalgebra is the Harper Hamiltonian \cite{Harper1955}
\begin{equation}\label{eq:harper}
H_{\mathrm{Harper}} = T_1 + T_1^* + T_2 + T_2^*.
\end{equation}
It is self-adjoint and its spectrum $\sigma(H_{\mathrm{Harper}})$ is a compact
subset of $\mathbb{R}$. The spectral properties depend crucially on whether
$\alpha_B$ is rational or irrational \cite{Hofstadter1976}.

\begin{itemize}
    \item \emph{Rational flux:} For rational flux \(\alpha_B=p/q\), the model admits a finite-dimensional
Bloch--Floquet reduction over the magnetic Brillouin zone, and the spectrum
is described by finitely many bands. Spectral gaps, when present, determine
Fermi projections in the magnetic noncommutative torus.
    \item \emph{Irrational flux:} The spectrum is a Cantor set (Hofstadter butterfly)
        with a dense set of gaps.
\end{itemize}
In either case, there are intervals of energy that belong to the resolvent set of
$H_{\mathrm{Harper}}$.

A \emph{spectral gap} is an open interval $(E_-,E_+)$ with $E_-<E_+$ such that
$(E_-,E_+)\cap\sigma(H_{\mathrm{Harper}})=\varnothing$. Choose a Fermi energy
$E_F$ inside a spectral gap. Then the \emph{Fermi projection} is defined by
\begin{equation}\label{eq:fermi-projection}
P_F = \mathbf{1}_{(-\infty,E_F)}(H_{\mathrm{Harper}}),
\end{equation}
where $\mathbf{1}$ denotes the spectral projection. Because
$E_F\notin\sigma(H_{\mathrm{Harper}})$, the characteristic function
$\mathbf 1_{(-\infty,E_F)}$ agrees on $\sigma(H_{\mathrm{Harper}})$ with a
continuous function; hence, by continuous functional calculus,
$P_F\in C^*(H_{\mathrm{Harper}})\subset C^*(T_1,T_2)$, the magnetic Brillouin
torus. In particular $P_F$ is a projection, $P_F^2 = P_F = P_F^*$.

Via the inclusion $C^*(T_1,T_2)=C^*(U_3,U_4)\hookrightarrow A^4_{\Theta_{\chi,L}}$,
the projection $P_F$ first defines a class in the $K_0$ of the magnetic two-torus,
whose image under the induced map is the class
$[P_F]\in K_0(A^4_{\Theta_{\chi,L}})$. Under the identification
$K_0(A^4_{\Theta_{\chi,L}})\cong\Lambda^{\mathrm{even}}\mathbb{Z}^4$, a class carried
by the $(U_3,U_4)$-subtorus has components only in degree zero and in the
$(3,4)$-degree part; hence the only possibly nonzero degree-two component is
$C_{34}$ (associated with the $\hat\Pi_x,\hat\Pi_y$ directions), while
$C_{12}=C_{13}=C_{14}=C_{23}=C_{24}=0$ and $Q=0$. The gap-label formula
\eqref{eq:gap-label-formula} then reduces to
\begin{equation}\label{eq:harper-gap-label}
\mathfrak G_{\chi,L}(P_F) = N - \frac{B_{\mathrm{in}}a_x a_y}{2\pi\hbar}\,C_{34},
\end{equation}
where $N$ is the ordinary trace (the filling factor) and $C_{34}$ is an integer
(the Hall conductance in units of $e^2/h$). This is the familiar integer quantum
Hall effect expressed in the toroidal shadow language.

Thus the magnetic Brillouin subalgebra example shows how a spectral gap and the
corresponding Fermi projection naturally produce a $K_0$ class whose degree-two
component $C_{34}$ captures the magnetic (Hall) contribution to the integrated
density of states. The same reasoning applies to any gapped Hamiltonian that
lives in the $T_1,T_2$ subalgebra.

\subsection{Direct-sum Powers--Rieffel projection}

Let $p_{12}$ be a Powers--Rieffel projection in the $(U_1,U_2)$-subtorus and
$p_{34}$ a Powers--Rieffel projection in the $(U_3,U_4)$-subtorus. Define
\[
P_{12,34}
=
\begin{bmatrix}
p_{12} & 0\\
0 & p_{34}
\end{bmatrix}
\in M_2(A^4_{\Theta_{\chi,L}}).
\]
Then
\[
P_{12,34}^2=P_{12,34},
\qquad
P_{12,34}^*=P_{12,34}.
\]
With the standard orientation convention,
\[
C_{12}(P_{12,34})=1,
\qquad
C_{34}(P_{12,34})=1,
\]
and
\[
Q(P_{12,34})=0.
\]
Thus $P_{12,34}$ illustrates nontrivial lower-degree gap-label components, but
it does not realize the degree-four component.

\subsection{Physical interpretation of the gap-label components}\label{sec:interpretation}

The components in $\mathfrak G_{\chi,L}(p)$ should be understood as
$K$-theoretic gap-label data. A topological-phase interpretation requires
realization by Fermi projections of gapped Hamiltonians and stability under
gap-preserving perturbations.

The component $C_{34}$ is associated with the magnetic Brillouin torus and
reduces to the usual integer quantum Hall Chern number in the appropriate
setting. The Hall conductance contribution is
\[
\sigma_{xy}=\frac{e^2}{h}C_{34}.
\]
The component $C_{12}$ records a configuration-space noncommutativity
contribution. Its coefficient depends on the spatial cell area, reflecting how
the chosen periodic cell samples spatial noncommutativity. The components
$C_{13}$ and $C_{24}$ are mixed position--momentum components. They are
invisible in the purely magnetic Brillouin subalgebra and motivate the use of
the full noncommutative $4$-torus.

The component $Q$ is the degree-four component of the $K_0$-class. It is absent
from the ordinary two-dimensional magnetic Brillouin torus. Its coefficient is
independent of the spatial cell area and is governed by the intrinsic
nondegeneracy combination $\hbar-B_{\mathrm{in}}\vartheta$. By the Chern-character
identification for smooth noncommutative tori, when a smooth Fermi projection
$P_F$ represents the class, the top-degree component is represented by the
noncommutative second Chern character
\[
\operatorname{Ch}_2(P_F)
=\kappa
\sum_{\rho\in S_4}(-1)^\rho\,
\tau_N\!\left(
P_F\delta_{\rho(1)}(P_F)\delta_{\rho(2)}(P_F)
\delta_{\rho(3)}(P_F)\delta_{\rho(4)}(P_F)
\right),
\]
where $\tau_N=\operatorname{Tr}_N\otimes\tau$ and the constant $\kappa$ depends
only on the chosen normalization of the Chern character. In the commutative
limit this is the second Chern number of the occupied projection over a
four-torus. Hence $Q$ is the natural degree-four invariant corresponding to the
same Chern-character component that quantizes nonlinear Hall responses in
four-dimensional quantum Hall systems \cite{ZhangHu2001,Price2015} and in
lower-dimensional topological-pump realizations \cite{Lohse2018}.

This Chern-character identification is used here only to name the $K$-theoretic
component: since the four toroidal directions are compactified position- and
momentum-like NCQM directions rather than ordinary spatial ones, $Q$ is treated
as a gap-label component and not as a response coefficient
(Section~\ref{sec:towards-q}).

\subsection{A finite-range model class for the degree-four component (motivation)}\label{sec:towards-q}

The results of this paper --- the gap-label formula \eqref{eq:gap-label-formula},
the area-independence of Proposition~\ref{prop:q-coefficient}, and the
separation results (Theorems~\ref{thm:trace-range} and~\ref{thm:pfaffian-separation}) --- are statements
about $K_0(A^4_{\Theta_{\chi,L}})$ and its trace pairing, and do not depend on
the material of this subsection. We include the following model class only to
indicate how the degree-four component could be realized by a Hamiltonian; we
do \emph{not} prove that $Q(P_F)\neq0$ for an intrinsic $\Theta_{\chi,L}$, and
we flag the precise open point below. To interpret $Q$ through Hamiltonians one
would construct a gapped
\[
H\in M_n(A^4_{\Theta_{\chi,L}})
\]
or an appropriate self-adjoint affiliated operator, together with a Fermi
level $E_F$ in a spectral gap, such that
\[
P_F=\mathbf 1_{(-\infty,E_F)}(H)
\]
satisfies
\[
Q(P_F)\neq 0.
\]

A useful model class is provided by the finite-range Wilson--Dirac family on
the noncommutative $4$-torus. Let
\[
S_j=\frac{U_j-U_j^*}{2i},\qquad
C_j=\frac{U_j+U_j^*}{2},\qquad j=1,\ldots,4,
\]
and let $\Gamma_1,\ldots,\Gamma_5$ be self-adjoint matrices satisfying the
Clifford relations
\[
\Gamma_i\Gamma_j+\Gamma_j\Gamma_i=2\delta_{ij}I.
\]
For $m\in\mathbb R$, define
\[
H_m=
\sum_{j=1}^4 S_j\otimes\Gamma_j
+
\left[m+\sum_{j=1}^4(1-C_j)\right]\otimes\Gamma_5.
\]
This Hamiltonian is finite range in the torus generators: it is a finite
Laurent polynomial in the unitaries $U_j$ and their adjoints. Thus it is the
noncommutative-torus analogue of a lattice tight-binding Hamiltonian with
nearest-neighbour hopping.

For $\Theta=0$, the above expression reduces to the standard commutative
four-dimensional lattice Dirac model. Its Bloch Hamiltonian is
\[
h_m(k)=
\sum_{j=1}^4\sin(k_j)\Gamma_j
+
\left[m+\sum_{j=1}^4(1-\cos k_j)\right]\Gamma_5,
\qquad k\in\mathbb T^4.
\]
Since
\[
h_m(k)^2=
\sum_{j=1}^4\sin^2(k_j)
+
\left[m+\sum_{j=1}^4(1-\cos k_j)\right]^2,
\]
the gap can close only at the usual discrete mass values. For mass parameters
in the nontrivial windows, the occupied-band projection of this commutative
model has nonzero second Chern number. Equivalently, its $K_0$-class has a
nonzero degree-four component.

For $\Theta\neq 0$, the same finite Laurent polynomial in the $U_j$ and $U_j^*$
defines a Hamiltonian $H_m(\Theta)$ over $A^4_\Theta$. To compare the $\Theta=0$
and $\Theta\neq0$ models rigorously we use Rieffel's continuous-fields framework
for twisted group $C^*$-algebras \cite{RieffelContinuousFields}: along the path
$s\mapsto s\Theta$ the cocycles $\sigma_{s\Theta}$ on $\mathbb Z^4$ vary
continuously, so the noncommutative tori $\{A^4_{s\Theta}\}_{s\in[0,1]}$ form a
continuous field of $C^*$-algebras, with $\Theta$ deformed continuously to $0$;
the generators $U_j(s)$ are continuous sections, and hence the finite Laurent
polynomial $H_m(s\Theta)$ depends norm-continuously on $s$.
The Fermi gap and the associated $K_0$-class are then controlled by the
following elementary but rigorous stability statement, applied along
$s\mapsto H_m(s\Theta)$. If the Fermi gap stays open for all $s\in[0,1]$, then
$Q(P_F(\Theta))=Q(P_F(0))$, which is nonzero for masses in the nontrivial
windows.

We emphasize what is and is not established. The homotopy invariance below is
rigorous. What we do \emph{not} prove is that the Fermi gap of
$H_m(s\Theta_{\chi,L})$ stays open for all $s\in[0,1]$ at the intrinsic
deformation $\Theta_{\chi,L}$ determined by a given sector. For $\|\Theta\|$
small relative to the commutative gap this follows from norm-continuity and
stability of isolated spectral projections; for a general intrinsic
$\Theta_{\chi,L}$ it is a genuine spectral-gap problem that we leave open. Thus
the realization $Q(P_F)\neq0$ for an intrinsic NCQM torus is a conjecture, not
a theorem of this paper.

The following elementary stability statement records the precise topological
input used here.

\begin{proposition}[Gap-preserving homotopy of Fermi projections]
Let $A=A^4_{\Theta_{\chi,L}}$, and let
\[
H(t)\in M_n(A),\qquad t\in[0,T],
\]
be a norm-continuous path of self-adjoint elements. Suppose that there exists
$E_F\in\mathbb R$ such that
\[
E_F\notin\sigma(H(t))
\]
for every $t\in[0,T]$. Then
\[
P_F(t)=\mathbf 1_{(-\infty,E_F)}(H(t))
\]
is a norm-continuous path of projections in $M_n(A)$. Consequently
\[
[P_F(t)]=[P_F(0)]\quad\text{in }K_0(A)
\]
for all $t\in[0,T]$. In particular, the components $N$, $C_{ij}$, and $Q$
of the $K_0$-class remain constant along the gap-preserving path.
\end{proposition}

\begin{proof}
The spectral-gap hypothesis permits a contour $\Gamma$ in the resolvent set
separating the occupied and unoccupied spectral parts of $H(t)$, locally
uniformly in $t$. By the Riesz formula,
\[
P_F(t)=\frac{1}{2\pi i}\int_\Gamma (z-H(t))^{-1}\,dz .
\]
Norm-continuity of $H(t)$ implies norm-continuity of the resolvent on
$\Gamma$, hence norm-continuity of $P_F(t)$. A norm-continuous path of
projections defines a homotopy in $K_0(A)$, so
\[
[P_F(t)]=[P_F(0)].
\]
The final assertion follows from the identification
$K_0(A^4_\Theta)\cong\Lambda^{\mathrm{even}}\mathbb Z^4$.
\end{proof}

This shows that the top-degree gap-label component is not merely formal: it is
naturally probed by finite-range Hamiltonians over the toroidal shadow algebra,
and the same homotopy principle would apply to a later Floquet or
adiabatic-pump construction defining a norm-continuous gap-preserving path over
$A^4_{\Theta_{\chi,L}}$. A response or pumping interpretation of $Q$ --- which
would require a specified adiabatic protocol, a gap-stability proof, and a
Kubo/Chern-character derivation --- is not part of this paper.

\section{Concluding remarks}\label{sec:concluding}

This paper distinguishes two operations often conflated in noncommutative
quantum mechanics --- linear transformations of represented operators (Bopp
shifts, Darboux normalizations) and unitary equivalence of
$\mathfrak g_{\mathrm{NC}}$-representations --- and turns the distinction into a
concrete separation criterion, computable in an arithmetic subfamily. The upshot is that the question
whether NCQM is ``merely'' quantum mechanics is not settled by the existence of a
Bopp shift: the shift is a functor that forgets the sector data $\vartheta$ and
$B_{\mathrm{in}}$, whereas the toroidal shadow retains a $K$-theoretic invariant
of them that is detectable and, under a mild arithmetic hypothesis, computable.

The results rigorously established here are the following. The central character
$(\hbar,\vartheta,B_{\mathrm{in}})$ is a unitary invariant of an irreducible
sector, so a Bopp shift is not a unitary equivalence; recorded functorially,
$\mathcal B_T$ collapses $(\hbar,\vartheta,B_{\mathrm{in}})\mapsto\hbar$ already
on objects (Section~\ref{sec:bopp-functor}). The toroidal shadow assigns to a
sector the noncommutative $4$-torus $A^4_{\Theta_{\chi,L}}$, whose trace pairing
gives the gap-label formula \eqref{eq:gap-label-formula}. The top-degree
coefficient equals $\operatorname{Pf}(\Theta_{\chi,L})=(\vartheta
B_{\mathrm{in}}-\hbar)/((2\pi)^2\hbar)$ and is independent of the cell area
(Proposition~\ref{prop:q-coefficient}); it is the area-independent top-degree
generator of the $K_0$-trace range. Since strong Morita equivalence preserves
the trace range up to positive scaling, non-proportionality of the trace ranges
obstructs Morita equivalence (Theorem~\ref{thm:trace-range}), separating
equal-$\hbar$ sectors that $\mathcal B_T$ identifies. In the arithmetic
subfamily where the dimensionless deformation coefficients $\vartheta/A,
B_{\mathrm{in}}A/\hbar$ are algebraic, the transcendence
of $1/2\pi$ forces the scale to be trivial (Lemma~\ref{lem:lambda-one}), and $|\operatorname{Pf}|$ becomes a
computable criterion: $|\vartheta B_{\mathrm{in}}-\hbar|\neq|\vartheta'
B'_{\mathrm{in}}-\hbar|$ already separates the sectors
(Theorem~\ref{thm:pfaffian-separation}), with the exact condition in
Proposition~\ref{prop:exact-morita}. The same combination $\hbar-\vartheta
B_{\mathrm{in}}$ that fixes the nondegeneracy of the sector thus records both the
degree-four gap label and the Pfaffian obstruction.

Two directions are left open, and are not claimed as results here. First, the
computable Pfaffian criterion is established in the algebraic-parameter
subfamily; for arbitrary real parameters the obstruction is the general (not
closed-form) non-proportionality of trace ranges, and we do not address Morita
equivalence across different $\hbar$ nor the finer
$\mathrm{SO}(4,4\,|\,\mathbb Z)$ arithmetic needed for a complete classification
\cite{RieffelSchwarz1999,Li2004}.

Second, and more substantively, there is the question of realizing the degree-four component \(Q\) dynamically through a mixed Kubo--Středa response over the toroidal shadow. Let \(j,a,b,c\) be four distinct toroidal directions, where \(j\) denotes the output direction, \(a\) the electric or adiabatically driven direction, and \((b,c)\) a complementary cocycle or twist plane. For a sufficiently regular, gap-preserving family of smooth Fermi projections \(P_F\) over the corresponding continuous field of noncommutative four-tori, the natural problem is to determine whether the variation, with respect to the \((b,c)\)-entry of the deformation matrix, of the degree-two Chern pairing in the \((j,a)\)-directions is governed by the degree-four Chern pairing and hence by \(Q(P_F)\). An adiabatic dual-action protocol may provide a realization of the electric component of this mixed response, while the independent deformation of the complementary twist entry supplies the Středa component. Establishing the precise normalization and analytic hypotheses, and proving the existence of a common Fermi gap for a concrete Hamiltonian at a prescribed intrinsic deformation, are left for future work. The finite-range Wilson--Dirac family of Section~\ref{sec:towards-q} provides a natural model class in which these questions may be investigated.

\section*{Acknowledgements}

The author is grateful to the anonymous reviewers, whose constructive and
positive feedback significantly improved the quality and rigor of this paper. He
also thanks his daughter, Nefertiti, for kindly providing him with pen and paper
during the calculations.

\end{document}